\documentclass{elsarticle}

\usepackage{amssymb}
\usepackage{amsthm, amsmath}
\usepackage[boxed, linesnumbered, ruled]{algorithm2e}

 \newtheorem{theorem}{Theorem}[section]
 
 \newtheorem{lemma}[theorem]{Lemma}
\newtheorem{definition}[theorem]{Definition}
\newtheorem{remark}[theorem]{Remark}

\newcommand{\proofofref}{}
\newproof{zproofof}{Proof of \proofofref}
\newenvironment{proofof}[1]
 {\renewcommand{\proofofref}{#1}\zproofof}
 {\endzproofof}

\newcommand{\set}[1]{\left\{#1\right\}}

\journal{arXiv.org}

\begin{document}

\begin{frontmatter}


\title{Maximal and maximum transitive relation contained in a given binary relation\tnoteref{label1}}
\tnotetext[label1]{A preliminary version of this article appeared in \textit{Maximal and Maximum Transitive Relation Contained in a Given Binary Relation}, COCOON 2015, Beijing, China}

\author[cmi]{Sourav Chakraborty}
\ead{sourav@cmi.ac.in}

\author[ju]{Shamik Ghosh}
\ead{sghosh@math.jdvu.ac.in}

\author[cmi]{Nitesh Jha\corref{cor1}}
\ead{nitesh.n.jha@gmail.com}

\author[cmi]{Sasanka Roy\fnref{fn1}}	
\ead{sasanka.ro@gmail.com}

\cortext[cor1]{Corresponding author}
\fntext[fn1]{Present Address: Indian Statistical Institute,  203 Barrackpore Trunk Road, Kolkata 700108, India}

%

\address[cmi]{Chennai Mathematical Institute, Chennai, India}

\address[ju]{Department of Mathematics, Jadavpur University, Kolkata, India}



\begin{abstract}
We study the problem of finding a \textit{maximal} transitive relation contained in a given binary relation. Given a binary relation of size $m$ defined on a set of size $n$, we present a polynomial time algorithm that finds a maximal transitive sub-relation in time $O(n^2 + nm)$. 

We also study the problem of finding a \textit{maximum} transitive relation contained in a binary relation. This is the problem of computing a maximum transitive subgraph in a given digraph. For the class of directed graphs with the underlying graph being triangle-free, we present a $0.874$-approximation algorithm. This is achieved via a simple connection to the problem of maximum directed cut. Further, we give an upper bound for the size of any maximum transitive relation to be $m/4 + cm^{4/5}$, where $c > 0$ and $m$ is the number of edges in the digraph.
\end{abstract}


\end{frontmatter}



\section{Introduction}
All relations considered in this study are binary relations. We represent a relation alternately as a digraph to simplify the presentation at places (see Section \ref{notation} for definitions). Transitivity is a fundamental property of relations. Given the importance of relations and the transitivity property, it is not surprising that various related problems have been studied in detail and have found widespread application in different fields of study. 

Some of the fundamental problems related to transitivity that have been long studied are - given a relation $\rho$, checking whether $\rho$ is transitive, finding the transitive closure of $\rho$, finding the maximum transitive relation contained in $\rho$, partitioning $\rho$ into smallest number of transitive relations. Various algorithms have been proposed for these problems and some hardness results have also been proved.

In this article, we study two related problems on transitivity. First - given a relation, obtain a \textit{maximal} transitive relation contained in it. It is straight-forward to see that this can be solved in poly-time, hence our goal is to do this as efficiently as possible. Second - given a relation, obtain a \textit{maximum} transitive relation contained it. This problem was proven to be NP-complete  in \cite{DBLP:conf/stoc/Yannakakis78}. Here our approach is to find approximate solutions.

The problem of finding a \textit{maximum} transitive relation contained in a given relation is a generalisation of well-studied hard problems. For the class of digraphs such that the underlying graph is triangle-free, the problem of computing a maximum transitive subgraph is the same problem as the MAX-DICUT problem (see Section \ref{section:maximum-transitive}).  MAX-DICUT has well known inapproximability results.

We can also relate it to a problem of optimisation on a 3SAT instance. We look at the relation as a directed graph $G = (V, E)$, where $|V| = n$. For every pair for distinct vertices $(i, j)$ in $V$, create a boolean variable $x_{ij}$. Consider the following 3SAT formula.
\begin{align*}
C = \bigwedge_{1 \leqslant i < j < k \leqslant n} (x_{ij} \vee \overline{x}_{ik} \vee \overline{x}_{kj}) 
\end{align*}
Let $C'$ be a formula derived from $C$ such that any literal with variable $x_{ij}$ is removed if $(i,j) \notin E$. It is easy to see that a solution to $C'$ represents a subgraph of $G$. Specifically, a solution to $C'$ is also transitive. To see this, observe that for every triplet $(i,j,k)$, if a clause $(x_{ij} \vee \overline{x}_{ik} \vee \overline{x}_{kj})$ is satisfied, then either the edge $(i,j)$ is included or at least one of the edges $(i,k)$ or $(k,j)$ is excluded. To get the maximum transitive subgraph, the solution must maximize the number of variables set to 1. To conclude, the maximum transitive subgraph problem is same as the problem of finding a satisfying solution to a 3SAT formula that also maximizes the number of variables assigned the value `true'.

\subsection{Our Results}
The usual greedy algorithm for finding a maximal substructure - satisfying a given property $\mathcal{P}$ - starts with the empty set and incrementally grows the substructure while maintaining the property $\mathcal{P}$. Finally it ends when the set becomes maximal. Thus checking for maximality is a subroutine for the usual greedy algorithm.

\subsection*{Maximal Transitive Subgraph}
We consider the problem of \textit{maximal transitive subgraph} -- output a transitive subgraph of maximal size (in terms of number of edges) contained in a given directed graph. Let's consider two related problems first. Let $G$ be a directed graph with $n$ vertices and $m$ edges. Given a transitive subgraph $S$ of $G$, can we add any more edges to $S$ and still maintain transitivity? We can check the maximality of $S$ in time $O(n^{w + 1})$ using a standard algorithm (where $O(n^w)$ is the complexity of multiplying two $n \times n$ matrices.) A related problem is -- given a transitive subgraph $S$ of $G$, compute a maximal transitive subgraph of $G$ that contains $S$. The naive algorithm takes $O(n^{w + 2})$ time. 

We give an algorithm that computes a maximal transitive subgraph in $O(n^2 + nm)$ time.  The interesting part of our algorithm is that we avoid checking for maximality explicitly but output is still maximal. This is the first such algorithm that improves upon the standard techniques which have a complexity of $O(n^{w + 1})$.

\begin{theorem} Let $D$ be a digraph with $n$ vertices and $m$ edges. Then there is an algorithm that given $D$, outputs a maximal transitive subgraph contained in $D$, in time $O(n^2 + nm)$.
\end{theorem}

We present the algorithms and the proof of correctness related to the following theorem in Section \ref{maximal}.

\subsection*{Maximum Transitive Subgraph}
We then study the \textit{maximum transitive subgraph} (MTS) problem -- compute a transitive subgraph of largest size contained in a given directed graph. This problem was proven to be NP-complete by Yannakakis in \cite{DBLP:conf/stoc/Yannakakis78}. We start by studying approximation algorithms for this problem.

The MTS problem is a generalization of well-studied hard problems. For the class of triangle-free graphs, the problem of finding a maximum transitive subgraph in a directed graph is the same problem as the MAX-DICUT problem.  MAX-DICUT has well known hardness and inapproximability results.

\ \\\noindent
\textbf{Approximation:} We give a simple 0.25-approximation algorithm of obtaining an MTS in a general graph. 

\begin{theorem}
There exists a poly-time algorithm to obtain an $m/4$ sized transitive subgraph in any directed graph $D$ with $m$ edges. This gives a $1/4$-approximation algorithm for maximum transitive subgraph problem.
\end{theorem}

For the case where the underlying undirected graph is triangle free, we give a 0.874-approximation for the MTS problem. The idea there is to look at the related problem of directed maximum cuts in the same graph. To the best of our knowledge, no approximation algorithms are present in the literature which present any ratio better than 0.25.

Let $UG(D)$ represents the underlying undirected graph of digraph $D$.

\begin{theorem}
There exists a 0.874-approximation algorithm for finding the maximum transitive subgraph in a digraph $D$ such that $UG(D)$ is triangle-free.
\end{theorem}

\ \\\noindent
\textbf{Upper Bound:} Another interesting questions is how large the MTS can theoretically be. In a triangle-free (underlying undirected) graph, we know that there is a one-to-one correspondence between directed-cuts and transitive subgraphs. We prove that in triangle free graphs with $m$ edges, any directed cut is of size at most $m/4 + cm^{4/5}$ for some $c > 0$. This gives the same bound for the size of an MTS. This also shows that the approach of finding MTS approximations via bipartite subgraphs can't have better constant approximation ratio than 1/4.

\begin{theorem}
For every $m$, there exists a digraph $D$ with $m$ edges such that $UG(D)$ is triangle-free and the size of any directed cut in $D$ is at most $m/4 + cm^{4/5}$ for some $c > 0$.
\end{theorem}

These results are described in Section \ref{section:maximum-transitive}.

\subsection{Related Results}

The transitive property is a fundamental property of binary
relations. Various important algorithmic problems with respect to
transitive property has been studied and used. One very important and
well studied problem is finding
the transitive closure of a binary relation $\rho$ (that is the smallest binary 
relation which contains $\rho$ and is transitive). This problem of
finding transitive closure has been studied way back in
1960s. Warshall~\cite{Warshall} gave an algorithm to find the
transitive closure is time $O(n^3)$, where $n$ is the size of the set
on which the binary relation is defined. Using different techniques
\cite{Purdom}  gave an $O(n^2 + nm)$ algorithms, where $m$ is the
number of elements in $\rho$. Modifying the algorithm of Warshall, Nuutila
\cite{nuutila} connected the problem of finding transitive closure
with matrix multiplication. With the latest knowledge of matrix
multiplication (~\cite{coppersmith} and \cite{virginia}) we can compute the transitive
closure of a binary relation on $n$ elements using $O(n^{2.37})$ time
complexity. 

Another important problem connected to transitive property is
the finding the transitive reduction of a binary relation.
Transitive reduction of a binary relation $\rho$ is the minimal
sub-relation whose transitive closure is same as the transitive closure
of $\rho$. This was introduce by Aho et al~\cite{Aho} and they also
gave the tight complexity bounds. A closely related concept to the
transitive reduction is the maximal equivalent graph, introduced by
Moyles~\cite{Moyles}.

Given a binary relation, partitioning it as a union of transitive
relations  is another very important related problem (see ~\cite{Pothen}). 
A plethora of work has been done on this problem in recent times as
this problem has found application in biomedical studies.

The Maximum Transitive Subgraph (MTS) problem has been studied in the field of parameterization.
Arnborg et~al.\ \cite{DBLP:conf/icalp/ArnborgLS88} showed that the problem of MTS is fixed parameter tractable. They give an alternate proof of Courcelle's theorem \cite{DBLP:journals/iandc/Courcelle90} and express the MTS problem in \textit{Extended Monadic Second Order}, thus giving a meta-algorithm for the problem. This algorithm is not explicit and $f$ is known to be only a computable function.

The MTS problem has been studied in a more general setting as the \textsc{Transitivity Editing}  problem where the goal is to compute the minimum number of edge insertions or deletions in order to make the input digraph transitive. Weller et al \cite{DBLP:journals/jcss/WellerKNU12} prove its NP-hardness and give a fixed-parameter algorithm that runs in time $O(2.57^k + n^3)$ for an $n$-vertex digraph if $k$ edge  modifications are sufficient to make the digraph transitive. This result also applies to the case where only edge deletions are allowed -- the MTS problem.

\section{Notations} \label{notation}
Let $S=\set{1,2,\ldots,n}$, where $n$ is a natural number. A binary relation $\rho$ on $S$ is a subset of the cross product $S\times S$. We only consider binary relations in this study. Any relation $\rho$ on $S$ can be represented by a $(0,1)$ matrix $A=(a_{ij})_{n\times n}$ of size $n\times n$, where  
$$a_{ij}=\left\{%
\begin{array}{ll}
1, & \mbox{if } (i,j)\in\rho\\
0, & \mbox{otherwise.}
\end{array}\right.$$
Similarly, a relation $\rho$ on $S$ can represented by a directed graph with $S$ as the vertex set and elements of $\rho$ as the arcs of the directed graph. 

In this paper we do not distinguish between a relation and its matrix representation or its directed graph representation.  So for a given relation $\rho$, if $(i,j) \in \rho$, we sometimes refer to it as the arc $(i,j)$ being present and sometimes as the adjacency matrix entry $\rho_{ij} = 1$.

If $\rho$ is a binary relation on $S$ then the size of $\rho$ (denoted by $m$) is the number of arcs in the directed graph
corresponding to $\rho$. In other words, it is the number of pairs
$(i,j) \in S\times S$ such that $(i,j) \in \rho$. 

If $\rho$ is a binary relation on $S$ we say $\rho'$ is contained in
$\rho$ (or is a sub-relation) if for all $i, j \in S$, $(i,j) \in
\rho'$ implies $(i,j) \in \rho$. 

\begin{definition} 
A binary relation $\rho$ on $S$ is called {\em transitive} if for all $a,b,c\in
S$, $(a,b) \in \rho, (b,c) \in \rho$ implies $(a,c)\in\rho$.  
\end{definition}

For a binary relation $\rho$ on $S$ a sub-relation $\alpha$ is said to be a \textit{maximal transitive} relation contained in $\rho$ if there does not exist any transitive relation $\beta$ such that $\alpha$ is strictly contained in $\beta$ and $\beta$ is contained in $\rho$. A \textit{maximum transitive} relation contained in $\rho$ is a largest relation contained in $\rho$.


\section{Maximal transitive relation finding algorithms}\label{maximal}

We first present an algorithm which finds a maximal transitive
relation contained in a given binary relation in $O(n^3)$ and then we
improve it to obtain another algorithm for this with time complexity
$O(n^2+mn)$. 

\subsection{$O(n^3)$ algorithm for finding maximal transitive sub-relation}

\vspace{-0.5cm}

\noindent
\IncMargin{0em}
\begin{algorithm}\label{algo1}
\SetKwData{Left}{left}\SetKwData{This}{this}\SetKwData{Up}{up}
\SetKwFunction{Union}{Union}\SetKwFunction{FindCompress}{FindCompress}
\SetKwInOut{Input}{Input}\SetKwInOut{Output}{Output}

\Input{An $n\times n$ matrix $A=(a_{ij})$ representing a relation.}
\Output{A matrix $T=(t_{ij})$ which is a maximal transitive sub-relation contained in $A$.}

\BlankLine
\For{$i\leftarrow 1$ \KwTo $n$}{
	\For{$j\leftarrow 1$ \KwTo $n$, $j \neq i$}{\label{forins}
		\If{$a_{ij}=1$}{
			\For{$k=1$ to $n$}{
				\If{$k\neq j$ and $a_{ik}=0$}{set $a_{jk}=0$}
				\If{$k\neq i$ and $a_{kj}=0$}{set $a_{ki}=0$}
			}
		}
	}
}
\Return A \\
\caption{Finding a maximal transitive sub-relation}
\end{algorithm}\DecMargin{1em}


\noindent
\begin{theorem}\label{th1}
\textbf{Algorithm 1} correctly finds a maximal transitive sub-relation in a given relation in time $O(n^3)$.
\end{theorem}

\begin{proof}
It is easy to see that the time complexity of the algorithm is
$O(n^3)$. For the proof of correctness, all we need to prove is that the output $T$ of the
algorithm is transitive and maximal. The transitivity of the output $T$
is proved in Lemma \ref{lem12} and the maximality of $T$ is proved in 
Lemma \ref{lem13}. 
\end{proof}

\subsection{Proof of Correctness of Algorithm 1}

Before we prove the correctness of \textbf{Algorithm 1}, let us
make some simple observations about the algorithm. In this section we
will treat the binary relation on a set $S$ as a directed graph with
vertex set $S$. So the
\textbf{Algorithm 1} takes a directed graph $A$ on $n$ vertices (labelled
$1$ to $n$) and outputs a directed transitive subgraph $T$ that is
maximal, that is, one cannot add arcs from $G$ to $T$ to obtain a
bigger transitive graph. 
In the algorithm, note that changing an entry $a_{ij}$ from 1 to 0 implies deletion of the arc $(i, j)$.

\begin{definition}
At any stage of the \textbf{Algorithm 1} we say the arc $(a,b)$ is
visited if at some earlier stage of the algorithm when $i = a$ in Line
$1$ and $j =b$ in Line $2$ we had $a_{ij} =1$. 
\end{definition}

\noindent
\begin{remark}\label{rem11}
We first note the following obvious but important facts of the
\textbf{Algorithm 1}:

\begin{enumerate}
\item[\textbf{(1)}] No new arc is created during the algorithm because it never
  changes an entry $a_{ij}$ in the matrix $A$ from $0$ to $1$.  It
  only deletes arcs. 

\item[\textbf{(2)}] Line $1$, $2$ and $3$ of the algorithm implies that the
  algorithm visits the arcs one by one (in a particular 
  order). And while visiting an arc it decides whether or not to
  delete some arcs.

\item[\textbf{(3)}] Since in Line $1$ the $i$ increases from $1$ to $n$ so the
  algorithm first visits the arcs starting from vertex $1$ and
  then the arcs starting from vertex $2$ and then the arcs starting
  from vertex $3$ and so on. 

\item[\textbf{(4)}] Arcs are deleted only in Line $6$ and Line $9$ in the
  algorithm. 

\item[\textbf{(5)}] While the for loop in Line $1$ is in the $i$-th iteration (that
  is when the algorithm is visiting an arc starting at $i$)
  no arc starting from the $i$ is deleted. 
  In Line $6$ only arcs starting from $j$ are deleted and $j \neq i$
  from Line $2$. And in Line $9$ only arcs ending in $i$ are deleted.

\item[\textbf{(6)}] In Line $2$ the condition $j\neq i$ is given just for ease of
  understanding the algorithm. As such even if the condition was not
  there the algorithm would have the same output because if $j = i$ in
  Line $2$ and the algorithm pass line $3$ (that is $a_{ii} = 1$) then 
  Line $6$ would read as ``if $a_{ik}=0$ write $a_{ik}=0$''  and Line
  $9$ would read as ``if $a_{ki}=0$ write $a_{ki}=0$'', both of which
  are no action statement. 

\item[\textbf{(7)}] Similarly, in Line $5$ the condition $k \neq j$ is given just for
  ease of understanding of the algorithm. If the condition was not
  there even then the algorithm would have produced the same result
  because from Line $3$ we already have $a_{ij}=1$ and thus if $k= j$
  then $a_{ik}=a_{ij}\neq 0$.

\item[\textbf{(8)}] Similarly, the condition $k \neq i$ in Line $9$ has no
  particular role in the algorithm.

\end{enumerate}
\end{remark}

One of the most important lemma for the proof of correctness is the following:

\noindent
\begin{lemma}\label{lem11}
An arc once visited in \textbf{Algorithm 1} cannot be deleted later on.
\end{lemma}

\begin{proof} Let us prove by contradiction. 

Suppose at a certain point in the algorithm's run the arc $(i,j)$ has
already been visited, and then when the algorithm is visiting some
other arc starting from vertex $r$ the algorithm decides to delete
the arc $(i,j)$.  

If such an arc $(i,j)$ which is deleted after being visited exists
then there must a first one also. Without loss of generality we can
assume that the arc $(i,j)$ is the first such arc: that 
is when the algorithm decides to delete the arc $(i,j)$ no other arc that has
been visited by the algorithm has been deleted.   

By point number $3$ in Remark~\ref{rem11},
$r\geq i$.  From point number $5$ in Remark~\ref{rem11} we can say
that $r \neq i$. So we have $r > i$. 

We now consider two cases depending on whether the algorithm decides
to delete the arc $(i,j)$ is Line $6$ or Line $9$.

\begin{figure}[h]
\begin{center}
\includegraphics[scale=0.4]{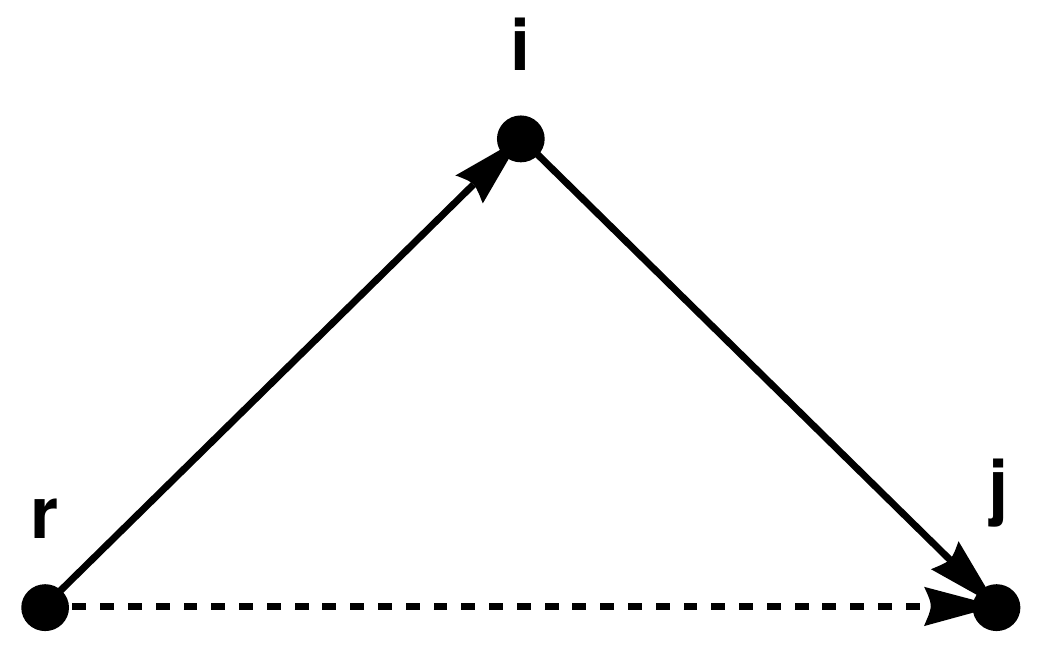}\hspace{1in} \includegraphics[scale=0.4]{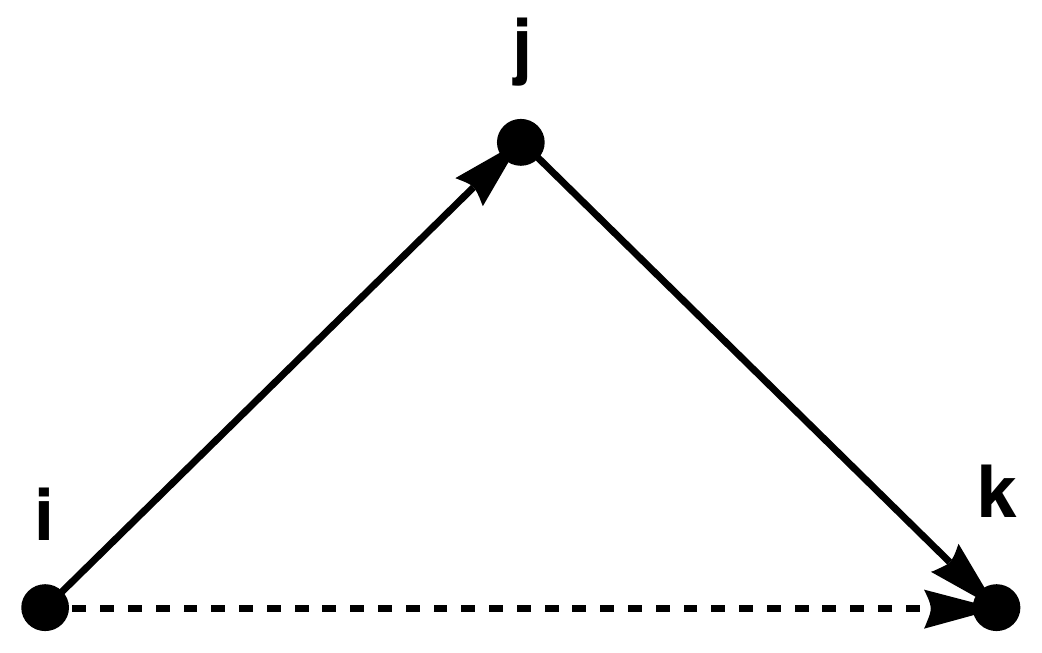}

\textbf{Case I\hspace{2.5in} Case II}
\end{center}
\caption{Diagrams of the two cases for Lemma~\ref{lem11}} \label{f:case2}
\end{figure}

\vspace{1em}
{\bf Case I.}\ \ Suppose $(i,j)$ is deleted in Line $6$, when the
algorithm was visiting an arc starting from vertex $r$. 
Since the algorithm is deleting $(i,j)$ in Line $6$ so from Line $3$
and Line $5$ we have, at that stage, $a_{ri} = 1$ and $a_{rj}=0$ (just
like in Figure~\ref{f:case2}(left)).

Since no arc is ever created by the algorithm (point $1$ in Remark \ref{rem11}),
$a_{ri}$ was $1$ when the arc $(i,j)$ was visited. 
So at the stage when the algorithm was visiting arc $(i,j)$, $a_{rj}$ must be $1$, 
otherwise $(r,i)$ would be deleted by Line $9$. Thus $(r,j)$ was
deleted after  visiting the arc $(i,j)$  and but before time $(i,j)$ is being deleted. 

By Remark~\ref{rem11}(5), $(r,j)$ cannot be deleted when
visiting an arc starting from $r$. So $(r,j)$ must have been deleted
when visiting an arc starting from vertex $r_1$ and $r_1<r$.  

We now split this case into two cases depending on whether
$r_1=j$ or $r_1 \neq j$.

\begin{figure}[h]
\begin{center}
\includegraphics[scale=0.4]{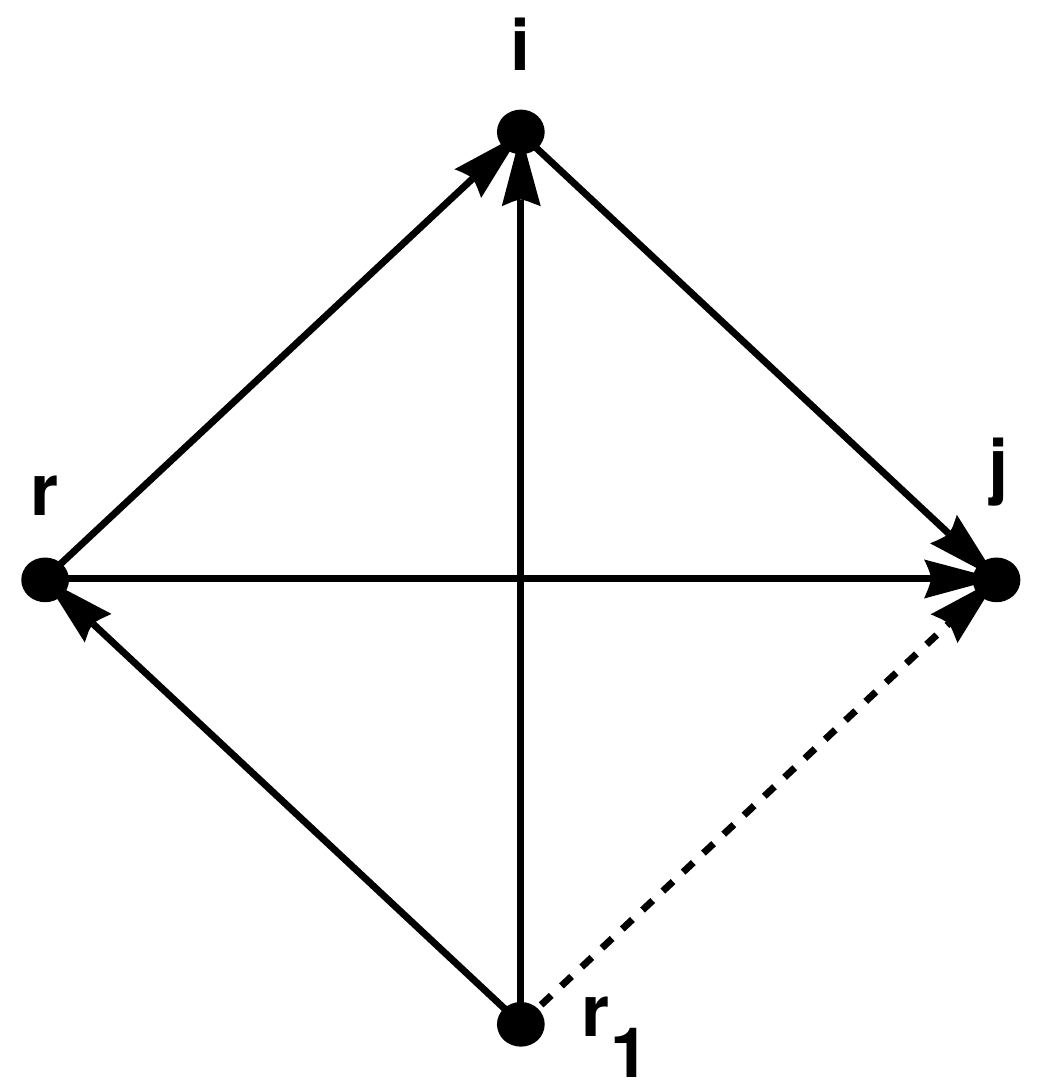}\qquad \includegraphics[scale=0.4]{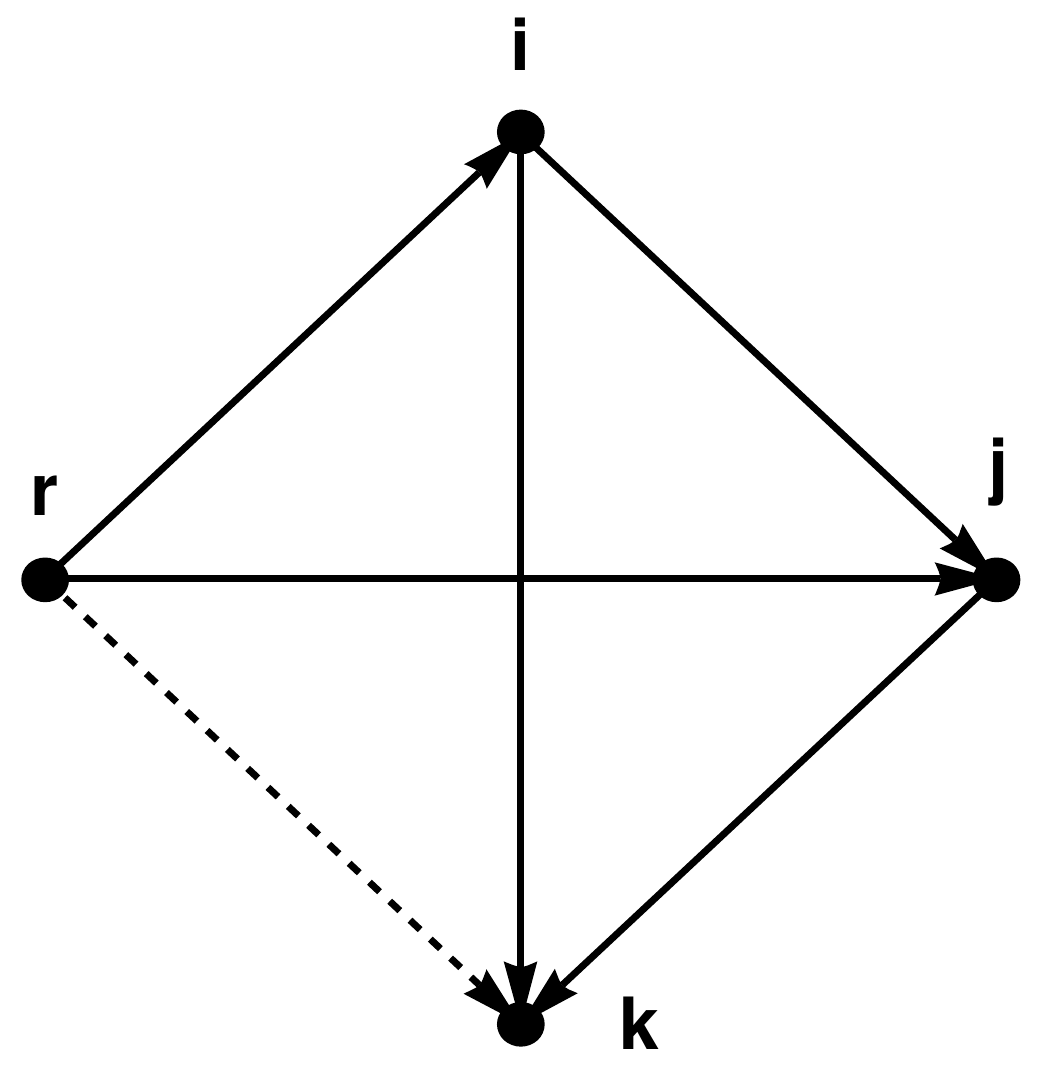}

\textbf{Case Ia\hspace{2.2in} Case Ib}
\end{center}
\caption{Diagram for subcases of Case 1 for Lemma~\ref{lem11}}\label{f:case1a}
\end{figure}

\

\textbf{Case Ia: ($r_1 \neq j$)}  

By Remark~\ref{rem11}(5) we know
at the set of arcs starting from vertex 
$r_1$ must have remained unchanged during the $r_1$-th iteration of
Line $1$.  

But since in the $r_1$-th iteration of Line $1$ the arc $(r,j)$ was deleted so 
$(r_1,r)$ must have been present while $(r_1,j)$ was absent. Also 
if $a_{r_1i} =0$ when visiting the arc $(r_1, r)$, the algorithm would have found
$a_{r_1r} = 1$ and $a_{r_1i} =0$ and in that case would have deleted
$(r,i)$ is Line $6$. That would contradict that fact the the arc
$(r,i)$ was present when the arc $(i,j)$ was being deleted. Thus 
at the start of the $r_1$-th iteration of Line $1$ the situation would
have been like in Figure~\ref{f:case1a}(left)).

But in that case, when visiting $(r_1,i)$ the algorithm would have
found $a_{r_1i} =1$ and  $a_{r_1j}=0$ and then would
have deleted the arc $(i,j)$. But by assumption the arc $(i,j)$ is deleted
when visiting arc $(r,i)$ and not an arc starting at $r_1$. So we
get a contradiction. And thus if $s \neq j$ we have a
contradiction. 

\vspace{1em} 
\textbf{Case Ib: ($r_1 = j$)}  

Let the arc $(r,j)$ be deleted when
the algorithm was visiting the arc $(r_1, k)$ (that is $(j,k)$) for
some $k$. Since the arc $(j,k)$ is deleted after the arc $(i,j)$ is
visited and before the arc $(r,i)$ is visited, so $i < j< r$. 

Now consider the stage when the arc $(j,k)$ is visited by the
algorithm. If arc $(i,j)$ is not present at that time then the arc
$(i,j)$ would have been deleted which would contradict the assumption
that the arc $(i,j)$ is deleted when the algorithm was visiting
$(r,i)$.  So just before the stage when the algorithm was visiting
arc $(j,k)$ the situation would have been like in Figure~\ref{f:case1a}(right)).

So the arc $(i,k)$ was present when the algorithm was visiting the
arc $(j,k)$. But since $i<j$ so the arc $(i,k)$ must have been
visited already. By the minimality condition that $(i,j)$ is the
first arc that is visited and then deleted and since the arc $(i,j)$
is deleted when visiting arc $(r, i)$, so when the algorithm just started
visiting the arc $(r,i)$ the arc $(i, k)$ must be present.  Also at
that stage the arc $(r,k)$ was absent as it was absent when visiting
the arc $(j, k)$ and $j<r$. So when the algorithm just started to
visit $(r,i)$ the situation would have been like in
Figure~\ref{f:case1a}(right)) except the arc $(r,j)$ would also have
been missing.

When the algorithm was visiting the arc $(j,k)$ the arc $(r,k)$
was not there. But when the algorithm visited the arc $(i,j)$ the
arc $(r,k)$ must have been there, else the arc $(r,i)$ would have
been deleted at that stage, which would contradict our assumption that
$(i,j)$ was deleted when visiting $(r,i)$. So the arc $(r,k)$ must
have been deleted after the arc $(i,k)$ was visited but before the
arc $(j,k)$ was visited.

If the arc is deleted when visiting some arc starting with $k$ then
it means that $i<k<j$. Now consider the stage when the algorithm was
visiting $(r,i)$. As described earlier the situation would have been like in
Figure~\ref{f:case1a}(right)) except the arc $(r,j)$ would also have
been missing. Since $k< j$ so the algorithm would have deleted $(i,k)$ before it
deleted $(i,j)$. And since the algorithm has also visited $(i,k)$
earlier so this contradicts the the minimality condition of $(i,j)$
being the first visited arc to be deleted.

The other case being the arc deleted when visiting the some arc
ending in $r$, say $(t,r)$, where $i<t<j$. Thus during the $t$-th iteration of Line
$1$ the arcs $(t,r)$ is present while the arc $(t, k)$ is
absent. Now, since in the $t$-th iteration the arc $(r,j)$ is not
deleted thus it means that the arc $(t, j)$ was present during the
$t$-th iteration of Line $1$. But in that case since arcs $(t, j)$ and $(j, k)$
are present while $(t,k)$ is not present the algorithm would have
deleted the arc $(j,k)$ in the $t$-th iteration of Line $1$, this contradicts the
assumption that the arc $(r,j)$ is deleted in the $j$-th iteration of 
Line $1$ when visiting the arc $(j,k)$. 

\

Thus the arc $(i,j)$ cannot be deleted by the algorithm in Line $6$
when visiting an arc starting from $r$.

\vspace{1em}
{\bf Case II.}\ \ Suppose $(i,j)$ is deleted in Line $9$, when the
algorithm was visiting an arc starting from vertex $r$. In this case
$j = r$. And since $r > i$ so $j>i$. Say the arc $(i,j)$ is deleted
when visiting arc $(j,k)$, for some vertex $k$.
Since the algorithm is deleting $(i,j)$ in Line $9$ so from Line $3$
and Line $8$ we have, at that stage, $a_{jk} = 1$ and $a_{ik}=0$   (cf. Figure
\ref{f:case2}(left)). 

Now if $a_{ik}$ was $0$ when the algorithm visited the arc $(i,j)$
then  the algorithm would have found $a_{ik} = 0$ and $a_{i,j} = 1$
and in that case would have deleted the arc $(j,k)$ in Line $6$. That
would give a contradiction as in a later stage of the algorithm (in
particular in the $j$-th iteration of Line $1$, with $j> i$) the arc $(j,k)$ is
present.  So when the arc $(i,j)$ was visited the arc $(i,k)$ was
present. 

Since by Remark~\ref{rem11}(5) the arc $(i,k)$ cannot be
deleted in the $i$th iteration of Line $1$, so the arc $(i,j)$ must have been
visited in the $i$-th iteration of Line $1$ and must have been deleted by the
algorithm at a later time but before the arc $(i,j)$ is deleted. 
But this would contradict the minimality of the arc $(i,j)$.

Hence even in this case also we get a contradiction. So this completes the
proof. 
\end{proof}

Next we prove that the output is transitive.

\noindent
\begin{lemma}\label{lem12}
The matrix $T$ output by the \textbf{Algorithm 1} is transitive.
\end{lemma}

\begin{proof}
Suppose $t_{ij}=1=t_{jk}$. By Remark \ref{rem11}(1) no arc
is created. So at all stages and in particular, at the initial stage
$a_{ij}=a_{jk}=1$. Suppose $a_{ik}=0$ at the initial stage. Then when
the algorithm visited $(i,j)$ or $(j,i)$ (whichever comes first), the
arc $(j,k)$ or $(i,j)$ (respectively) will be deleted for the lack of
the arc $(i,k)$, as $a_{ij}=a_{jk}=1$ throughout (cf. Figure
\ref{f:ttrans}). 

Thus suppose the arc $(i,k)$ is deleted at some stage, say, $r$-th
iteration of Line $1$. Now $r>i, j$ for otherwise the arc $(i,k)$
would be deleted before the $i$-th or $j$-th iteration of Line $1$. And
in that case in the $i$-th or $j$-th iteration of Line $1$ (depending on which
of $i$ and $j$ is smaller) of Line $1$  either $(j,k)$ or $(i,j)$
would be deleted.  And then at the end at least one of $t_{ij}$ and
$t_{ik}$ must be $0$.

But then the arc $(i,k)$ is deleted during the $i$-th iteration of Line $1$ (as $i<r$). Since no arc is deleted once it is visited by Lemma \ref{lem11}, we have $t_{ik}=1$. Therefore $T$ is transitive. 
\end{proof}

\begin{figure}[h]
\begin{center}
\includegraphics[scale=0.4]{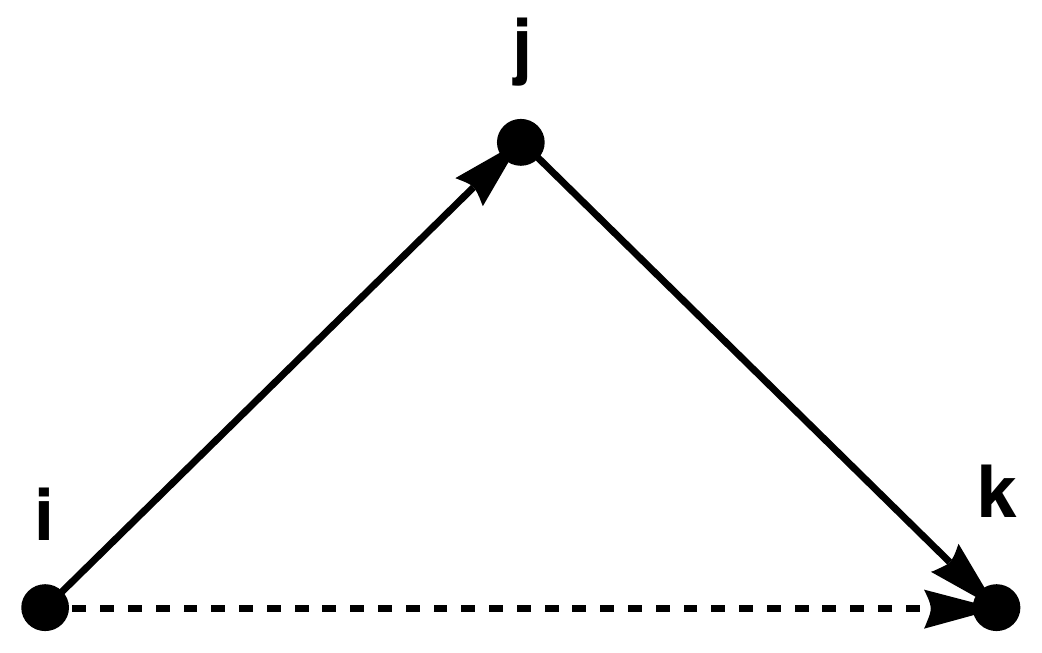}

\caption{Diagram for Lemma~\ref{lem12}}\label{f:ttrans}
\end{center}
\end{figure}

Using Lemma~\ref{lem12} and Lemma~\ref{lem11} we can finally prove the
correctness of the algorithm. 

\noindent
\begin{lemma}\label{lem13}
The matrix $T$ output by the \textbf{Algorithm 1} is a maximal transitive relation contained in $A$.  
\end{lemma}

\begin{proof}
$T$ is transitive by Lemma \ref{lem12}. Also by Remark~\ref{rem11}(1) the output matrix is contained in $A$. So the only
thing remaining to prove is that the output matrix $T$ is maximal. 

Now if $T$ is not a maximal transitive sub-relation then there must be
some arc (say $(a,b)$) such that the transitive closure of
$T\cup \{(a,b)\}$ is also contained in $A$. 

Now by Lemma \ref{lem11}, an arc once visited can never be
deleted. Also the algorithm is visiting every undeleted arc. 
Thus $T$ is the collection of visited arcs and these arcs
are present at every stage of the algorithm. 

Thus, every arc in the transitive closure of
$T\cup \{(a,b)\}$ that is not in $T$ must have been deleted in some
iteration of Line $1$. Let $(i,j)$ be the first arc to be deleted
among all the arcs that are in the of transitive closure of
$T\cup \{(a,b)\}$ but not in $T$. 

Clearly the transitive closure of $T\cup \{(i,j)\}$ is also contained
in $A$, and all the arcs in the transitive closure of $T\cup
\{(i,j)\}$  either is never deleted or is deleted after the arc
$(i,j)$ is deleted. Suppose the arc $(i,j)$ is deleted in the $r$-th 
iteration of Line $1$. We have $r\neq i$ by Remark~\ref{rem11}(5) and by
Lemma~\ref{lem11} we have $r<i$.  

We now consider two cases depending on whether $r$ is $j$ or not.

\begin{figure}[h]
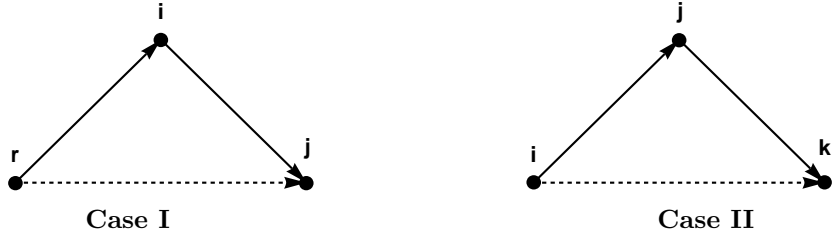

\begin{center}
\includegraphics[scale=0.4]{case1.pdf}\hspace{1in} \includegraphics[scale=0.4]{case2.pdf}

\textbf{Case I\hspace{2.5in} Case II}
\end{center}
\caption{Diagrams of the two cases for Lemma~\ref{lem13}}\label{f:max}
\end{figure}

\vspace{1em}
{\bf Case I:\ \ $r\neq j$ }

In this case, since the arc $(i,j)$ was deleted in the $r$-iteration
of Line $1$, the arc $(i,j)$ must have been deleted when the
algorithm was visiting the arc $(r,i)$. So at the stage when the arc
$(i,j)$ was deleted, the arc $(r,j)$ must not have been there (else
the algorithm wouldn't have deleted the arc $(i,j)$). 

If $a_{rj}=0$ in $A$, then $t_{rj}=0$ (by Remark
\ref{rem11}(1)). But by Lemma~\ref{lem11} $t_{ri}=1$ as the arc $(r,i)$
is being visited. So $T\cup\set{(i,j)}$ is not transitive (cf. Figure
\ref{f:max}(left)), and the transitive closure of $T\cup\set{(i,j)}$
must contain the arc $(r,j)$. Thus $a_{rj}=1$ in $A$, but the arc
$(r,j)$ is deleted in some stage of the algorithm but before the visit of the $r$-th
iteration of Line $1$, say, at $r_1$-th iteration of Line $1$, with $r_1<r$. 

Thus the arc $(r,j)$ is in the transitive closure of
$T\cup\set{(i,j)}$ and it got deleted before the deletion of arc
$(i,j)$. This is a contradiction to the fact that the arc $(i,j)$ was
the first arc to be deleted. So when $r\neq j$ we have a
contradiction.

\vspace{1em}
{\bf Case II:\ \ $r= j$ }

In this case, since the arc $(i,j)$ was deleted in the $j$-iteration
of Line $1$, the arc $(i,j)$ must have been deleted when the
algorithm was visiting some arc $(j,k)$, for some vertex $k$. 
So at the stage when the arc $(i,j)$ was deleted, the arc $(i,k)$
must not have been there (else 
the algorithm wouldn't have deleted the arc $(i,j)$). 

If $a_{ik}=0$ in $A$, then $t_{ik}=0$ (by Remark
\ref{rem11}(1)). But by Lemma~\ref{lem11} $t_{jk}=1$ as the arc $(j,k)$
is being visited. So $T\cup\set{(i,j)}$ is not transitive (cf. Figure
\ref{f:max}(right)), and the transitive closure of $T\cup\set{(i,j)}$
must contain the arc $(i,k)$. Thus $a_{ik}=1$ in $A$, but the arc
$(i,k)$ is deleted in some stage of the algorithm but before the visit of the $j$-th
iteration of Line $1$, say, at $r_1$-th iteration of Line $1$, with $r_1<j$. 

Thus the arc $(i,k)$ is in the transitive closure of
$T\cup\set{(i,j)}$ and it got deleted before the deletion of arc
$(i,j)$. This is a contradiction to the fact that the arc $(i,j)$ was
the first arc to be deleted. So when $r= j$ we have a
contradiction.

Since in both the case we face a contradiction so we have that the
output $T$ is a maximal transitive relation contained in $A$. 
\end{proof}

\subsection{Better running time analysis of Algorithm 1}

If we do a better analysis of the running time of the
\textbf{Algorithm 1} we can see that the algorithm has running time
$O(n^2 + nm)$. To see it more formally consider a new pseudocode of
the algorithm that we present as Algorithm~\ref{algo2}. It is not hard to see that both the algorithms are
basically same. 


\noindent
\IncMargin{0em}
\begin{algorithm}\label{algo2}
\SetKwData{Left}{left}\SetKwData{This}{this}\SetKwData{Up}{up}
\SetKwFunction{Union}{Union}\SetKwFunction{FindCompress}{FindCompress}
\SetKwInOut{Input}{Input}\SetKwInOut{Output}{Output}

\Input{An $n\times n$ matrix $A=(a_{ij})$ representing a binary relation.}
\Output{A matrix $T=(t_{ij})$ which is a maximal transitive relation contained in the given binary relation $A$.}

\BlankLine
\For{$i\leftarrow 1$ \KwTo $n$}{
Initialize $B_i = \emptyset$\\
\For{each $s\leftarrow 1$ \KwTo $n$, $j \neq i$}{\label{forins}
\If{$a_{ij}=1$}{ Include $j$ in $B_i$}} 
\For{ each $j \in B_i$}{
\For{ each $k=1$ to $n$}{
\If{$k\neq j$ and $a_{ik}=0$}{Make $a_{jk}=0$}
\If{$k\neq i$ and $a_{kj}=0$}{Make $a_{ki}=0$}
}
}
}

\caption{Finding a maximal transitive sub-relation}\label{algo1}
\end{algorithm}\DecMargin{1em}


\begin{theorem}
Algorithm \ref{algo2} correctly finds a maximal transitive relation contained in a given binary relation in $O(n^2+mn)$, where $m$ is the number of $1$'s in $A$.
\end{theorem}

\begin{proof}
The proof for correctness is same as in Theorem \ref{th1}. We
calculate only the time complexity of the algorithm and it is given by

\begin{eqnarray*}
&& \sum\limits_{i=1}^n (n+k_in), \mbox{ (where $k_i$ is the number of $1$'s in the $i^{\mbox{th}}$ row)}\\
& = & n^2+n\sum\limits_{i=1}^n k_i   =  n^2+mn. 
\end{eqnarray*}
\end{proof}


\section{Maximum Transitive Relation}\label{section:maximum-transitive}

In this section, we study the problem of obtaining a maximum transitive relation contained in a binary relation. We will be using the notation of directed graphs for binary relations. As before, let's assume that input directed graph has $m$ edges. Denote by $UG(D)$ the underlying graph of digraph $D$.

First, we state a well known result from graph theory.

\begin{lemma}\label{lem:maximum-trans1}
There exists a bipartite subgraph of size $m/2$ in any graph with $m$ edges. 
\end{lemma}

Obtaining such a bipartite graph deterministically in poly-time is a folklore result. This gives the following.

\begin{theorem}\label{theorem:one-fourth-approx}
There exists a poly-time algorithm to obtain an $m/4$ sized transitive subgraph in any directed graph $D$ with $m$ edges. This gives a $1/4$-approximation algorithm for maximum transitive subgraph problem. 
\end{theorem}
\begin{proof} From Lemma~\ref{lem:maximum-trans1}, we get a bipartite subgraph of $UG(D)$ of size at least $m/2$.  Now consider the original orientations on this bipartite subgraph. We collect all the edges in the direction that has more number of edges. This set of arcs is of size at least $m/4$ and is transitive as there are no directed paths of length two in the set.
\end{proof}

The obvious question is -- given a digraph with $m$ edges, is there a transitive subgraph of size $tm$ such that $t > 1/4$? We claim that this is not possible. We start by proving the following theorem.

\begin{theorem}\label{theorem:dicut-main}
For every $m$, there exists a digraph $D$ with $m$ edges such that $UG(D)$ is triangle-free and the size of any directed cut in $D$ is at most $m/4 + cm^{4/5}$ for some $c > 0$.
\end{theorem}

We later observe that there is a one-to-one correspondence between directed cuts and transitive subgraphs in a digraph. Hence, obtaining a transitive subgraph of size better than $m/4$ (in the constant multiple) would contradict this theorem - since this would break the upper bound on the size of any directed cut.

The \textit{max-cut problem} is an extensively well studied problem both in terms of finding good approximation algorithms and estimating its bounds combinatorially. Both its undirected and directed versions are NP-complete. Here we give an upper bound on the size of \textit{directed max-cut} using probabilistic arguments.

The following notation is borrowed from \cite{DBLP:journals/jgt/AlonBGLS07}. Let $G$ be an undirected graph and $U, V$ be a partition of the vertex set of $G$. A \textit{cut} $(U, V)$ is the set of edges with one endpoint in $U$ and other endpoint in $V$. Call $e(U, V)$ the size of cut $(U, V)$. Define
\[
 f(G)  = \max_{(U,V)} e(U, V) \qquad \text{and,} \qquad f(m) = \min_{G : |E(G)| = m} f(G)
\]

Finding a max-cut was proved to be NP-complete in \cite{DBLP:journals/tcs/GareyJS76}. Goemans and Williamson give a semidefinite programming based algorithm in \cite{DBLP:journals/jacm/GoemansW95} to achieve an approximation ratio of $0.878$. Under the Unique Games Conjecture, this is the best possible \cite{DBLP:journals/siamcomp/KhotKMO07}. But a $0.5$-approximation algorithm is straight forward - randomly put each vertex in $U$ or $V$, leading to an expected cut size of $m/2$. Hence, $f(m) \ge m/2$. Various bounds have been proposed for $f(m)$, most notably in \cite{DBLP:journals/jct/ErdosFPS88, DBLP:journals/combinatorica/Alon96}. Following is an upper bound for $f(m)$ in triangle free graph.

\begin{theorem}[Alon \cite{DBLP:journals/combinatorica/Alon96}]\label{lemma:alon-1} There exists  a constant $c' > 0$ such that for every $m$ there exists a triangle-free graph $G$ with $m$ edges satisfying $f(G) \le m/2 + c' m^{4/5}$.
\end{theorem}

Let $H$ be a directed graph and $U, V$ be a partition of vertex set of $H$. A \textit{cut} of $H$ is similarly defined as before. A \textit{directed cut} $(U, V)$ is the set of edges with starting point in $U$ and ending point in $V$. Call $e(U, V)$ the size of cut $(U, V)$. Define
\[
 g(H)  = \max_{(U,V)} e(U, V) \qquad \text{and,} \qquad g(m) = \min_{H : |E(H)| = m} g(H)
\]

Finding a directed cut of maximum size is NP-complete (via a simple reduction from the max-cut problem). \cite{DBLP:journals/jacm/GoemansW95} gave a $0.796$ approximation for this problem. Again, a $0.25$ approximation is simple, given the $0.5$-approximation of max-cut. Since it is easy to find a cut of size $m/2$ in undirected graphs and a directed cut of size $m/4$ in directed graphs, an obvious question is how much better one can do as a fraction of $m$. Alon proved in \cite{DBLP:journals/combinatorica/Alon96} that the factor $1/2$ can not be improved for max-cuts. We prove that the factor $1/4$ can't be improved for directed max-cuts.

We now prove the following bound. For any $m$, there exists a directed graph with $m$ edges such that for some $c > 0$,
\[
g(H) \le m/4 + c m^{4/5}
\]

The proof idea is as follows. From Theorem \ref{lemma:alon-1}, we know that for every $m$ there exists an undirected graph with $m$ edges, all whose cuts are bounded by $m/2 + o(m)$ in size. For any given $m$ in our case, we start with the undirected graph of Theorem \ref{lemma:alon-1} satisfying the above bound. We orient this graph uniformly at random. We then prove that every cut of size more than $m/4$ will be highly balanced, in the sense that - the cut will have almost the same number of edges going from left to right and right to left. We formalise these ideas below.

We define a notion of balanced cuts of a directed graph and balanced directed graphs.
\begin{definition}[$\delta$-balanced cut]
For a directed graph, consider a cut $(U,V)$. The cut is $\delta$-balanced if 
\[
| e(U, V) - e(V, U) | \le \delta \left( \frac{e(U, V) + e(V, U)}{2} \right)
\]
\end{definition}

\begin{definition}[$(k,\delta)$-balanced graph]
A directed graph $H$ is $(k, \delta)$-balanced if every cut of $H$ of size at least $k$ is $\delta$-balanced.
\end{definition}

\begin{lemma}\label{lemma:graph-with-balanced-cuts}
For any $m$,  $\delta > 0$ and $k \le m$, there exists a directed graph $H$ on $n$ vertices and $m$ edges such that $H$ is $(k, \delta)$-balanced if $n < k \delta^2 / 6$.
\end{lemma}

\begin{proof}
For the given $m$, we start with an undirected graph $G$ satisfying the condition in Theorem \ref{lemma:alon-1}. We orient the edges of $G$ uniformly at random and independently and call it $H$. 

Let $C = (U,V)$ be a cut in the undirected graph $G$ of size at least $k$. We first calculate the probability (over the random orientations of $G$) that $C$ is not $\delta$-balanced in $H$. 
\begin{align}
P[C \text{ is not } \delta\text{-balanced}]  
&= P[| e(U, V) - e(V, U) | > \delta (e(U, V) + e(V, U))/2] \\
&= 2 P[e(U,V) > (1 + \delta) |C|/2]  \label{eq:chernoff}
\end{align}

For each edge $e_i$ in the cut $(U,V)$, define a random variable $X_i$ as follows,
\begin{equation*}
X_i = 	\begin{cases}
	1	& \text{if $e_i$ is directed from $U$ to $V$} \\
	0	& \text{otherwise}
	\end{cases}
\end{equation*}
$X_i$'s are i.i.d. random variables with probability $1/2$. Then, $e(U,V) = \sum_{e_i \in (U,V)} X_i$ with mean $|C|/2$. We apply the standard Chernoff bound to get an upper bound for the probability in Equation (\ref{eq:chernoff}), 
\begin{align*}
P[C \text{ is not } \delta\text{-balanced}]  
&\le 2 \exp (- \delta^2 |C| / 6) \\
&\le 2 \exp (- \delta^2 k / 6) 
\end{align*}

We now calculate the probability that the graph $H$ is $(k, \delta)$-balanced.
\begin{align*}
P[H \text{ is } & (k, \delta)\text{-balanced}] \\
&= 1 - P[\text{there exists a cut $C$ in $H$ of size at least $k$ which is not $\delta$ balanced}] \\
&= 1 - P\left[\bigcup_{\text{cut } C,  |C| \ge k}{\text{$C$ is not $\delta$-balanced}} \right] \\
& \ge 1 - 2^n (2 \exp (- \delta^2 k / 6)) \\
& > 0, \quad \text{ if } n < k \delta^2 / 6
\end{align*}
\end{proof}

\noindent
The following lemma gives us a directed graph $H$, such that any cut of size at least $m/4$ in $H$ is `well' balanced. This result is used in proving the final theorem.
\begin{lemma}\label{lemma:graph-with-specific-cuts}
For any $m$, there exists a directed graph $H$ with $m$ edges such that $H$ is $(m/4, \alpha/m^{1/5})$-balanced for some $\alpha > 0$.
\end{lemma}
\begin{proof}
By choosing $k = m/4$ and $\delta =  \alpha/m^{1/5}$ in Lemma \ref{lemma:graph-with-balanced-cuts}, we get $H$ if 
\begin{align*}
n &\le (m/4) (\alpha/m^{1/5})^2 /6 \\
\text{or, } m &\ge (24 / \alpha^2)^{5/3} n^{5/3}
\end{align*}
The counterexample in Theorem \ref{lemma:alon-1} requires that $m = (1/8 + o(1)) n^{5/3}$. Hence we need $\alpha \ge \sqrt{24}/(1/8 + o(1))^{3/10}$.
\end{proof}

\noindent
We now prove our main claim.

\begin{proofof}{Theorem \ref{theorem:dicut-main}}
For the given $m$, Lemma \ref{lemma:graph-with-specific-cuts} gives a digraph $H$ that is $(m/4, \alpha/m^{1/5})$--balanced, which would imply that every cut of size at least $m/4$ is $\alpha/m^{1/5}$--balanced. Consider any cut $(U,V)$ in $H$. We have,
\begin{align*}
|e(U,V) - e(V,U)| &\le \alpha / m^{1/5} |(U,V)| \\
|e(U,V)|, |e(V,U)| &\le |(U,V)|/2 + (\alpha / m^{1/5}) |(U,V)|/4 \\
&\le (m/2  + c' m^{4/5})/2 + (\alpha / 4m^{1/5}) m \\
& \le m/4 + (c'/2 + \alpha/4) m^{4/5}
\end{align*}

In the second last inequality we use the fact that $|(U,V)| \le m/2  + c' m^{4/5}$ from Theorem \ref{lemma:alon-1}. This completes the proof with the choice of $c =  (c'/2 + \alpha/4)$.


\end{proofof}

%
%
%
%

In order to improve upon the approximation factor, we focus on the class of \textit{triangle-free} directed graphs. First we make the following simple observation about triangle-free directed graphs.

\begin{lemma}\label{lemma:no-length-two-path}
Given a digraph $D$ such that $UG(D)$ is triangle-free; any transitive subgraph of $D$ has no directed paths of length two.
\end{lemma}

Let $G$ be a digraph and $U$, $V$ be a partition of the vertex set of $H$. A \textit{directed cut} $(U,V)$ is the set of edges with a starting in $U$ and ending point in $V$. The MAX-DICUT problem is the problem of obtaining a largest directed cut in a graph. This is NP-hard. \cite{DBLP:conf/ipco/LewinLZ02} gives an approximation algorithm for the MAX-DICUT problem.

\begin{theorem}[see \cite{DBLP:conf/ipco/LewinLZ02}]
There exists a 0.874-approximation algorithm for the MAX-DICUT problem.
\end{theorem}

As a corollary of Lemma \ref{lemma:no-length-two-path}, we have the following.
\begin{lemma}\label{lemma:directed-cut-equal-transitive-sub}
In a digraph $D$ such that $UG(D)$ is triangle-free, every directed cut  of $D$ is also a transitive subgraph of $D$.
\end{lemma}

This implies that finding the maximum transitive subgraph is same as the MAX-DICUT problem for digraphs $D$ with $UG(D)$ being triangle free.

\begin{theorem}
There exists a 0.874-approximation algorithm for finding the maximum transitive subgraph in a digraph $D$ such that $UG(D)$ is triangle-free.
\end{theorem}


\section{Conclusion}

We have presented an algorithm that given a directed graph on $n$ vertices and $m$ arcs outputs a maximal transitive subgraph is time $O(n^2 + nm)$. This is the first algorithm for finding maximal transitive subgraph that we know of, that does better than the usual greedy algorithm. Although it might be the case that this is an optimal algorithm, we are unable to prove a lower bound for this problem. 

There are many related problems for which one might expect similar kind of algorithm - that is $O(n^3)$ time algorithm that does better than the usual greedy algorithm. We would like to present them as open problems: 

\begin{enumerate} 
\item Given a directed graph $G$ on $n$ vertices and a transitive subgraph $H$ of $G$, check if $H$ is a maximal transitive subgraph of $G$.

\item Given a directed graph $G$ on $n$ vertices and a subgraph $H$ of $G$, find a maximal transitive subgraph of $G$ that also contains $H$.
\end{enumerate} 

Obviously an algorithm for the second problem would also give an algorithm for the first problem.

In the case of maximum transitive subgraph, the central question is obtaining a better approximation ratio than 1/4 in a general digraph.

\section{Acknowledgement}
This research was funded by the Chennai Mathematical Institute, H1, SIPCOT IT Park, Siruseri, Kelambakkam, 603103, India.

\bibliographystyle{elsarticle-num}  
\bibliography{transitive}




\end{document}